\documentclass[letterpaper, 12pt, onecolumn]{article}  
 
\usepackage[left=2.6cm, right=2.6cm, top=3cm, bottom=2cm]{geometry}

\usepackage{balance} 
\usepackage{xcolor,calc}

\usepackage[pdftex]{graphicx}
\usepackage{amsmath} 
\usepackage{amssymb}  
\usepackage[noadjust]{cite}
\usepackage{color}
\usepackage{enumerate}
\usepackage{multirow}
\usepackage{mathtools}
\usepackage{booktabs}
\usepackage{amsthm}
\usepackage{epstopdf}
\usepackage{algorithm,algcompatible}

\usepackage[normalem]{ulem}

\usepackage{url}
\usepackage{bm}

\usepackage{booktabs}

\newtheorem{theorem}{Theorem}
\newtheorem{remark}{Remark}
\newtheorem{proposition}{Proposition}

\makeatletter
\newsavebox\myboxA
\newsavebox\myboxB
\newlength\mylenA

\newcommand*\xoverline[2][0.75]{%
    \sbox{\myboxA}{$\m@th#2$}%
    \setbox\myboxB\null
    \ht\myboxB=\ht\myboxA%
    \dp\myboxB=\dp\myboxA%
    \wd\myboxB=#1\wd\myboxA
    \sbox\myboxB{$\m@th\overline{\copy\myboxB}$}
    \setlength\mylenA{\the\wd\myboxA}
    \addtolength\mylenA{-\the\wd\myboxB}%
    \ifdim\wd\myboxB<\wd\myboxA%
       \rlap{\hskip 0.5\mylenA\usebox\myboxB}{\usebox\myboxA}%
    \else
        \hskip -0.5\mylenA\rlap{\usebox\myboxA}{\hskip 0.5\mylenA\usebox\myboxB}%
    \fi}
\makeatother

\newcommand{\R}{\mathbb{R}}

\usepackage{tabularx}
\usepackage{booktabs}

\newtheorem{definition}{Definition}
\newtheorem{assumption}{Assumption}

\makeatletter
\let\NAT@parse\undefined
\makeatother

\title{Exploiting Model Sparsity in Adaptive MPC: \\ A Compressed Sensing Viewpoint}

\author{Monimoy Bujarbaruah$^*$, Charlott Vallon\thanks{Authors contributed equally to this work. The authors are with the MPC Lab, UC Berkeley, USA. E-mails: \tt\scriptsize{\{monimoyb,charlott\}@berkeley.edu} 
}
}

\begin{document}

\maketitle
  \thispagestyle{empty}
\pagestyle{empty}

\begin{abstract}
This paper proposes an Adaptive Stochastic Model Predictive Control (MPC) strategy for stable linear time-invariant systems in the presence of bounded disturbances. We consider multi-input, multi-output systems that can be expressed by a Finite Impulse Response (FIR) model. The parameters of the FIR model corresponding to each output are unknown but assumed \emph{sparse}. We estimate these parameters using the Recursive Least Squares algorithm. The estimates are then improved using set-based bounds obtained by solving the Basis Pursuit Denoising \cite{chen2001atomic} problem. Our approach is able to handle hard input constraints and probabilistic output constraints. Using tools from distributionally robust optimization, we reformulate the probabilistic output constraints as tractable convex second-order cone constraints, which enables us to pose our MPC design task as a convex optimization problem. 
The efficacy of the developed algorithm is highlighted with a thorough numerical example, where we demonstrate performance gain over the counterpart algorithm of \cite{bujarbaruahAdapFIR}, which does not utilize the sparsity information of the system impulse response parameters during control design.
\end{abstract}
\section{Introduction}
If the true dynamics of a system are uncertain, Adaptive Control strategies have been applied for meeting control objectives and ensuring system stability \cite{krstic1995nonlinear, ioannou1996robust, sastry2011adaptive}, typically under no output and input constraints. The uncertainty in these systems can be primarily attributed to two factors: $(i)$ model uncertainty (e.g. modeling mismatch and inaccuracies), and $(ii)$ exogenous disturbances (e.g. sensor noise). More recently, utilizing tools from classical Adaptive Control, Adaptive Model Predictive Control (MPC) \cite{tanaskovic2014adaptive, lorenzenAutomaticaAMPC, bujarbaruahAdapFIR, kohlerNonlAMPC} has established itself as a promising approach for control of uncertain systems subject to input and output constraints. 
For linear systems specifically, the literature on Adaptive MPC has extensively focused on either robust or probabilistic satisfaction of such imposed 
constraints on the system, using either state-space or input-output modeling. 

In \cite{soloperto2018learning, bujarArxivAdap} additive model uncertainties are considered with known system matrices, and imposed state and input constraints are robustly satisfied for all such realizable uncertainties. In \cite{lorenzenAutomaticaAMPC, hernandez2018stabilizing, kohler2019linear} robust state-input constraint satisfaction is extended to include both unknown system dynamics matrices and additive disturbances. The approach introduced in \cite{bujarArxivAdap} is also suited for satisfaction of probabilistic chance constraints on system states. Furthermore, the work in \cite{hewing2017cautious, koller2018learning, ostafew2014learning} uses Gaussian Process (GP) regression for real-time learning of an uncertain model and satisfies probabilistic state constraints with a traditional stochastic MPC \cite{farina2016stochastic} controller. Although such state-space modeling based Adaptive MPC controllers have proven to be effective, they involve construction of positive invariant sets \cite{kolmanovsky1998theory, blanchini1999set}, which can become computationally cumbersome. As a consequence, input-output modeling of systems has been opted in literature for proposing computationally efficient Adaptive MPC algorithms for certain applications (e.g. for stable, slow systems).

Adaptive MPC algorithms using input-output modeling of the system are presented in \cite{tanaskovic2014adaptive, 2017arXiv171207548T, bujarbaruahAdapFIR, FIRRoySmith19}, both for robust and probabilistic satisfaction of imposed input-output constraints. The  works  of  \cite{tanaskovic2014adaptive, 2017arXiv171207548T, FIRRoySmith19} deal with modeling errors in the Finite Impulse Response  (FIR)  domain, in the presence of a bounded additive disturbance, and prove  recursive  feasibility and stability \cite[Chapter~12]{borrelli2017predictive}  of the proposed robust Adaptive MPC approaches. These ideas are extended in \cite{bujarbaruahAdapFIR}, where a recursively feasible adaptive stochastic MPC algorithm is presented, demonstrating satisfaction of probabilistic output constraints and hard input constraints. The proposed approach in \cite{bujarbaruahAdapFIR} obtains a better performance compared to \cite{tanaskovic2014adaptive} measured in terms of closed-loop cost, owing to the allowance of output constraint violations with a certain (low) probability. 

In this paper, we build on the work of \cite{tanaskovic2014adaptive, bujarbaruahAdapFIR}, and propose an \emph{Adaptive Stochastic MPC} algorithm that considers probabilistic output constraints and hard input constraints for a Multi Input Multi Output (MIMO) system. Similar to \cite{bujarbaruahAdapFIR}, we consider an uncertain FIR model of the system that is subject to bounded disturbances with known mean and variance. The support for the set of all possible models, which we call the Feasible Parameter Set (FPS), is adapted at each timestep using a set membership based approach. In contrast to previous work \cite{bujarbaruahAdapFIR}, we additionally consider that the impulse response parameters corresponding to each output are sparse. Such sparse impulse response modeling can be motivated by \cite{benesty2006adaptive, etter1985identification} for MIMO systems. Our goal is to utilize this additional sparsity information to demonstrate performance improvement over the algorithm in \cite{bujarbaruahAdapFIR}. Our contributions can be summarized as follows:
\begin{itemize}
\item Offline before the control process, we compute a set containing all possible values of the unknown sparse FIR vectors corresponding to each output, with a very \emph{high} probability. This set, which we call the Feasible Sparse Parameter Set (FSPS) is computed using the Basis Pursuit Denoising \cite{chen2001atomic} problem.    

\item Online during the control process, we obtain a point estimate of the unknown system inside the intersection of the FPS and the FSPS, using a Recursive Least Squares (RLS) estimator. Using this estimated system, we propagate our nominal predicted outputs used in the MPC controller objective function to improve performance. Simultaneously, we ensure satisfaction of the output chance constraints for the unknown true system.

\item Through numerical simulations, we demonstrate that our algorithm exhibits better performance than the algorithm presented in \cite{bujarbaruahAdapFIR}.
\end{itemize}

\section{Problem Description}\label{sec:problem_des}

\subsection{System Modeling and Control Objective}\label{sec:model}
We consider stable linear time-invariant systems described by a Finite Impulse Response (FIR) model of the form
\begin{align} \label{sysmodel}
y(t) = H_a \Phi(t) + w(t),
\end{align}
where the number of inputs and outputs considered is $n_u$ and $n_y$, respectively. The FIR regressor vector of length $m$ is denoted by $\Phi(t) \in \mathbb{R}^{n_um} = [u_1(t-1), \dots, u_1(t-m),\dots,u_{n_u}(t-1),\dots,u_{n_u}(t-m)]^\top$, where $u_i(t)$ denotes the $i^{\mathrm{th}}$ input at time $t$. The matrix $H_a \in \mathbb{R}^{n_y \times n_um}$ is a matrix comprising of the impulse response coefficients that relate inputs to the outputs of the system. The disturbance vector $w(t) \in \mathbb{R}^{n_y}$ is assumed to be a zero-mean random variable with a known variance, component-wise bounded as
\begin{align}\label{disturbance_bound}
    \vert w_j(t)\vert \leq \bar w_j, \forall j = 1,2,\dots n_y,
\end{align}
where the $\bar w_j$ are assumed known. Finally, $y(t)\in\mathbb R^{n_y}$ is the measured output of the system.

Our goal is to control the output $y(t)$ while satisfying input and output constraints of the form
\begin{subequations} \label{constraints}
\begin{align}
C u(t) & \leq g, \quad t=0,1,\ldots, \label{contr_con}\\
\mathbb{P}\{E y(t) \leq p\} & \geq 1-\epsilon, \quad  t=0,1,\ldots, \label{eq:cc1} 
\end{align}
\end{subequations}
where $\epsilon\in(0,1)$ is the maximum allowed probability of output constraint violation. Following \cite{bujarbaruahAdapFIR}, we consider a single linear output chance constraint, meaning $E$ is a row vector and $p\in\R$. Notice that joint (linear) chance constraints \eqref{eq:cc1} can be reformulated into a set of individual (linear) chance constraints using Bonferroni's inequality, and can therefore also be addressed by our proposed framework.


\begin{assumption}\label{assump:1}
We assume that each row of impulse response matrix $H_a$ is sparse. Without loss of generality, we further assume that the sparsity index for each row, i.e. the number of nonzero entries, is at most $\bar{k}$. 
\end{assumption}
\subsection{Method Outline}
We assume in this paper that the system matrix $H_a$ in \eqref{sysmodel} is unknown. This paper proposes a method for identifying this unknown system matrix $H_a$, and using the estimate in a robust control formulation to safely regulate the constrained uncertain system. Our proposed method uses the following steps:
\begin{enumerate}
    \item \underline{Offline before the control process begins}: Use $q$ number of collected input sequence regressors $[\Phi_1, \Phi_2, \dots, \Phi_q]$ to compute a set $\mathcal{B}(H_a)$, called the Feasible Sparse Parameter Set (FSPS). We compute this set via the Basis Pursuit Denoising problem \cite{chen2001atomic}. The FSPS contains the true unknown model $H_a$, with a \emph{high} probability.
    
    \item \underline{Online during the control process}: At each timestep $t$, 
    \begin{enumerate}
        \item Obtain the current output measurement $y(t)$ and, using the known disturbance bounds \eqref{disturbance_bound}, update the time-varying set $\mathcal{F}(t)$, which we call the Feasible Parameter Set (FPS). The FPS is a set guaranteed to contain the true model $H_a$.
    
    \item Use the previous applied control inputs and measured outputs to construct an estimate $\mu_a(t)$ of $H_a$, lying in the intersection of $\mathcal{F}(t)$ and the offline-computed FSPS. The estimate is constructed using the Recursive Least Squares method. 
    
    \item Compute the input sequence that minimizes the objective function obtained with $\mu_a(t)$ while satisfying the input and output constraints \eqref{constraints} for all models in the FPS $\mathcal{F}(t)$. Apply the first computed control input and continue to step 2a.
    \end{enumerate}
    
\end{enumerate}

In the following Section~\ref{sec:backg}, we discuss steps 1-2(b). Step 2(c) is detailed in Section~\ref{control_synth}.


\section{Model Estimation and Adaptation}\label{sec:backg}
We approximate system \eqref{sysmodel} with the form
\begin{equation}\label{approxModel}
y(t)  = H(t) \Phi(t) + w(t),
\end{equation}
where our model $H(t) \in \mathbb{R}^{n_y \times n_um}$ is a random variable whose support we estimate online during the control process from the output measurements. The support for the set of all possible models $H(t)$ consistent with the recorded system data, which we call the Feasible Parameter Set (FPS), is guaranteed to contain the true model $H_a$. Based on the knowledge that system \eqref{sysmodel} has sparse impulse response properties, we also construct a Feasible Sparse Parameter Set (FSPS) offline by solving the Basis Pursuit Denoising Problem. During control run-time, a point estimate of $H_a$ is then computed to lie in the intersection of the offline-computed FSPS and the online-updated FPS. This estimate is then used in the control design. We decouple the offline and online phases of this design process and delineate the steps in detail in the following sections.  

\subsection{Offline: Construct the Feasible Sparse Parameter Set}\label{sec: fsps}
The Feasible Sparse Parameter Set (FSPS), denoted by $\mathcal{B}(H_a)$, is a function of the (unknown) true system response $H_a$, and is synthesized \emph{offline} utilizing the sparsity aspect of system responses from Assumption~\ref{assump:1}. Proposition~\ref{prop:sparse_bound} clarifies how this set is synthesized. We first introduce the following definition.

\begin{definition}[Restricted Isometry Property (RIP) \cite{candesTao}]
A matrix $\mathcal{A}$ satisfies the Restricted Isometry Property (RIP) of order $\bar{k}$, with constant $\delta \in [0, 1)$, if
\begin{align*}
    (1-\delta) \Vert x \Vert_2^2 \leq \Vert \mathcal{A}x \Vert_2^2 \leq (1+\delta) \Vert x \Vert_2^2,~\forall x~\textrm{$\bar{k}$ sparse}.
\end{align*}
The order-$\bar{k}$ restricted isometry constant $\delta_{\bar{k}}(\mathcal{A})$ is the smallest number $\delta$ such that the above inequality holds. 
\end{definition}

\begin{proposition}\label{prop:sparse_bound}
Suppose we collect $q$ output measurements offline. Suppose $\mathbf{y}_i = \mathcal{A}H^\top_{a_i} + \mathbf{w}_i $ for $i = 1,2,\dots,n_y$, where each $\mathbf{y}_i \in \mathbb{R}^{q \times 1}$ and $\mathcal{A} = [\Phi_1, \Phi_2, \dots, \Phi_q]^\top \in \mathbb{R}^{q \times n_u m}$, $\mathbf{w}_i \in \mathbb{R}^{q \times 1}$ with $\Vert \mathbf{w}_i \Vert_2 \leq \sqrt{q}\bar{w}_i$, and $H_{a_i} \in \mathbb{R}^{1 \times n_u m}$ denotes the $i^\mathrm{th}$ row of $H_a$. If $\delta_{2\bar{k}}(\mathcal{A}) < \sqrt{2}-1$, then any solution $\hat{\mathbf{x}}$ to the Basis Pursuit Denoising optimization problem
\begin{equation}\label{eq:sparse_recover}
    \begin{array}{llll}
        \displaystyle\min & \displaystyle \Vert \mathbf{x} \Vert_1 \\
        \text{s.t.} & \Vert \mathcal{A} \mathbf{x} - \mathbf{y}_i \Vert_2 \leq \sqrt{q}\bar{w}_i, ~(\textrm{denoted as}~\mathcal{B}(H_{a_i})) 
    \end{array}
\end{equation}
satisfies $\Vert \hat{\mathbf{x}} - H_{a_i} \Vert_2 \leq \bar{C}\sqrt{q}\bar{w}_i$ for some constant $\bar{C} \in \mathbb{R}$, and for all $i = 1,2,\dots,n_y$.
\end{proposition}
\begin{proof}
See \cite[Theorem~3.5.1]{YiMaBook}. The proof follows from the known result of \cite[Theorem~1.1]{candes2006stable}, with relaxation of $\delta_{4\bar{k}}(\mathcal{A}) < \frac{1}{4}$ to $\delta_{2\bar{k}}(\mathcal{A}) < \sqrt{2}-1$ .
\end{proof}
The set $\mathcal{B}(H_a)$ is obtained as $\mathcal{B}(H_a) = [\mathcal{B}(H_{a_1}), \dots, \mathcal{B}(H_{a_{n_y}})]^\top$. 
Note that this set $\mathcal{B}(H_a)$ is synthesized offline, as the regressor vectors $[\Phi_1, \Phi_2, \dots, \Phi_q]$ are required to come from a Gaussian distribution in order to ensure the RIP property of each matrix $\mathcal{A}_i$ for $i=1,2,\dots,n_y$. 
Such Gaussian inputs are not always allowable during control with system constraints \eqref{constraints}. Details on our choice of these offline regressors $[\Phi_1, \Phi_2, \dots, \Phi_q]$ to ensure $\delta_{2\bar{k}}(\mathcal{A}) < \sqrt{2}-1$ with \emph{high} probability are described in the Appendix. 
\begin{remark}
As pointed out in the proof of \cite[Theorem~3.5.1]{YiMaBook}, a possible choice of the numerical constant $\bar{C}$ in Proposition~\ref{prop:sparse_bound} is $\bar{C} = \frac{2}{\sqrt{\lambda}}$, with $\sqrt{\lambda} = \frac{1-\delta_{2\bar{k}}(1+\sqrt{2})}{\sqrt{2(1+\delta_{2\bar{k}})}}$. However, since $\delta_{2\bar{k}}$ is not \emph{exactly} known, computing $\bar{C}$ and hence $\mathcal{B}(H_a)$ accurately is not possible. We see that $\bar{C} \rightarrow 2\sqrt{2}$ as $\delta_{2\bar{k}} \rightarrow 0$. Therefore offline regressor vectors $[\Phi_1, \Phi_2, \dots, \Phi_q]$ should be chosen to ensure $\delta_{2\bar{k}} < \bar{\delta}$, with $\bar{\delta} \ll \sqrt{2}-1$. Under such choice of the offline regressors as shown in the Appendix, we pick the constant $\bar{C} \approx 2\sqrt{2}$.    
\end{remark}

\subsection{Online: Update the Feasible Parameter Set}\label{sec:FPSU}
Following \cite{tanaskovic2014adaptive, bujarbaruahAdapFIR}, a set-membership identification method is used for updating the time-varying FPS, denoted by $\mathcal{F}(t)$. The initialization of $\mathcal F(0)$ is done considering the fact that the true system (\ref{sysmodel}) is stable. We update the FPS as given by
\begin{multline}\label{support}
    \mathcal{F}(t)= \mathcal{F}(t-1) \cap \{H(t): H(t) \Phi(t)  \leq y (t) + \bar w\} \cap \{H(t): -H(t) \Phi(t) \leq -y(t) + \bar w \},
\end{multline}
where $\bar w = [\bar w_1,\ldots,\bar w_{n_y}]^\top$ is the bound of the additive disturbance given by \eqref{disturbance_bound}. 
In order to bound the computational complexity of \eqref{support} over time, an alternative algorithm to compute \eqref{support} is presented in \cite{tanaskovic2014adaptive}.

\subsection{Online: Obtain Point Estimate $\mu_a(t)$}\label{sec: kf}
We rewrite \eqref{approxModel} as
\begin{align*}
y(t) & = \bm{\Phi}(t)\mathbf{H}(t) + w(t),
\end{align*}
where $\bm{\Phi}(t) \in \mathbb{R}^{n_y \times n_yn_um}$ and $\mathbf{H}(t) \in \mathbb{R}^{n_yn_um \times 1}$ are shown in the Appendix. Furthermore, let $\sigma^2_w$ be the variance of the disturbance $w(t)$. Let the initial prior mean and variance estimates for true system be $\mu_{\mathbf{a}}(0)$ and $\sigma^2_{\mathbf{a}}(0)$ respectively. Now, the conditional mean and variance estimates, given measurements up to $y(t)$, can be obtained using the Recursive Least Squares method \cite[Sec.~(3.1)]{anderson1979}. These estimates may differ from the true conditional distribution parameters, as $w(t)$ is not necessarily assumed to be Gaussian.

We ensure that the mean point estimate $\mu_\mathbf{a}(t)$ at any time instant $t$ is chosen as a point contained in a set $\mathcal{F}_p(t)$, that is, $\mu_a(t) \in \mathcal{F}_p(t)$, with  
\begin{align}\label{eq:point_estimDom}
    \mathcal{F}_p(t) = \mathcal{F}(t) \cap \mathcal{B}(H_a).
\end{align}
One way of obtaining such a constrained estimate of $H_a$ in $\mathcal{F}_p(t)$ is to project the mean. As shown in \cite{bujarbaruahAdapFIR}, this is achieved by solving 
\begin{align}\label{eq:proj}
    \mu_{\mathbf{a}}(t)= \arg\min_{X\in \mathcal{F}_p(t)} (X-\mu_{\mathbf{a}}(t))^\top M (X-\mu_{\mathbf{a}}(t)),
\end{align}
where $M = (\sigma^2_{\mathbf{a}}(t))^{-1} \succ 0$. The mean in matrix form, that is, $\mu_a(t) \in \mathbb{R}^{n_y \times n_um}$ is obtained by reorganizing $\mu_{\mathbf{a}}(t) \in \mathbb{R}^{n_yn_um \times 1}$ into $n_um$ columns. This provides the minimum mean squared error estimate (with any linear estimator) of the true system $H_a$. Note that \eqref{eq:proj} is a convex optimization problem, as the set $\mathcal{F}_P$ in \eqref{eq:point_estimDom} is an intersection of two convex sets. We can further obtain a polytopic $\mathcal{F}_p$ by outer approximation of $\mathcal{B}(H_a)$ with the infinity norm.

\begin{remark}\label{rem:diff_with_acc18}
In \cite{bujarbaruahAdapFIR} the set $\mathcal{F}_p(t)$ in \eqref{eq:point_estimDom} is set as the FPS $\mathcal{F}(t)$ for all timesteps $t \geq 0$. Thus \cite{bujarbaruahAdapFIR} does not utilize the sparsity information of rows of $H_a$ to construct the nominal point model estimate $\mu_a(t)$. The construction of $\mathcal{B}(H_a)$ in \eqref{eq:point_estimDom} with our proposed approach utilizes such sparsity information from Assumption~\ref{assump:1}. However, this comes with additional offline computation of $\mathcal{B}(H_a)$, as shown in Proposition~\ref{prop:sparse_bound}. 
\end{remark}


\section{Control Synthesis}\label{control_synth}
\subsection{Prediction Model}\label{sec:prob}

Let $N>m$ be the prediction horizon for a predictive controller for system \eqref{sysmodel}. We denote the predicted system outputs at time $t$ by $y(k|t) = H(t) \Phi(k|t) + w(k)$, for some $H(t)\in \mathcal{F}(t)$. Similarly, $\Phi(k|t)$ denotes the \emph{predicted regressor vector}, for $k\in \{t+1,\dots,t+N\}$, and is computed as
\begin{equation}\label{phi_hor}
    \Phi(k|t)= W\Phi(k-1|t) + Z u(k-1|t), 
\end{equation}
where, the matrices $W$ and $Z$  are as reported in the Appendix. With these predicted regressor vectors, the estimated system  $\mu_a(t)$ obtained in Section~\ref{sec: kf} is used to propagate the nominal predicted outputs as $\hat{y}(k+1|t) = \mu_a(t) \Phi(k+1|t)$, for all $k \in \{t,\dots,t+N-1\}$. This is shown in the optimization problem presented in Section~\ref{sec:mpc_prob}. 

\subsection{Reformulation of Chance Constraints}
Within each prediction horizon we enforce $\mathbb P \{ E y(k|t) \leq p \} \geq 1-\epsilon$, where $y(k|t)$ is a function of some $H(t) \in \mathcal F(t)$. Therefore, to ensure satisfaction of (\ref{eq:cc1}) despite uncertainty in the true system, we must satisfy the constraint for all $H(t) \in \mathcal{F}(t)$. Using the theory of distributionally robust optimization \cite{calafiore2006distributionally, zymler2013distributionally}, we can conservatively approximate the output chance constraints \eqref{eq:cc1} as
\begin{align}\label{socp1}
\kappa_{\epsilon} \sqrt{{\bar{\Phi}}^\top(k|t)\Gamma{\bar{\Phi}}(k|t)} + \Phi^\top(k|t)\bar{E}\mathbf{H}(t) -p \leq 0, ~~ \forall H(t) \in \mathcal{F}(t),
\end{align}
where we have $k\in \{t+1,\dots,t+N\}$, $\kappa_{\epsilon} =\sqrt{\frac{1-\epsilon}{\epsilon}}$ and $\bar \Phi(k|t) = [\Phi^\top(k|t) \quad 1 \quad 1]^\top$. Here, $\Gamma \succeq 0$ is an appended covariance matrix shown in the Appendix. As $\mathcal{F}(t)$ is a polytope, \eqref{socp1} can be succinctly written as
\begin{equation}\label{socp}
    \kappa_{\epsilon} \sqrt{{\bar{\Phi}}^\top(k|t)\Gamma{\bar{\Phi}}(k|t)} + \Phi^\top(k|t)\bar{E}f^j(t) -p \leq 0,
\end{equation}
where $f^j(t)$ denote the vertices of the polytope $\mathcal{F}(t)$.

\begin{remark}\label{rem:fps_robust}
The reformulated chance constraints \eqref{socp} are \emph{robustly} satisfied for all $H(t) \in \mathcal{F}(t)$, and only the choice of the point model estimate $\mu_a(t)$ depends on the set $\mathcal{B}(H_a)$, as shown in \eqref{eq:point_estimDom}. Hence, satisfaction of \eqref{eq:cc1} is guaranteed by \eqref{socp}, despite $\mathbb{P}(H_a \in \mathcal{B}(H_a)) \neq 1$. 
\end{remark}

\subsection{MPC Problem}\label{sec:mpc_prob}
We solve the following optimization problem for given weight matrices $Q \in \mathbb{R}^{n_y \times n_y},S \in \mathbb{R}^{n_u \times n_u} \succ 0$: 
\begin{equation}\label{mpc_problem}
    \begin{array}{llll}
        \displaystyle\min_{U(t)} & \displaystyle \sum_{k=t}^{t+N-1}  [\hat{y}^\top(k|t)Q\hat{y}(k|t)+u^\top(k|t)Su(k|t)] + \hat{y}^\top(t+N|t)Q\hat{y}(t+N|t) \\
        \text{s.t.} & \hat{y}(k+1|t) = \mu_a(t)\Phi(k+1|t), \\[0.7ex]
         & \hat{y}(t|t) = y(t), \\[0.7ex]
        & Cu(k|t) \leq g, \\[0.7ex]
        & \Phi(t+N|t)  =W\Phi(t+N|t)+Z u(t+N-1|t), \\ [0.7ex]
        & \kappa_{\epsilon} \sqrt{{\bar{\Phi}}^\top(k+1|t)\Gamma{\bar{\Phi}}(k+1|t)} + \Phi^\top(k+1|t)\bar{E}f^j(t)\leq p, \\ [0.7ex]
        & \forall k = t, \ldots, t+N-1, \\ [0.7ex]
        & \forall f^j(t) \in \textrm{vertex}(\mathcal{F}(t)),~\mu_a(t) \in \mathcal{F}(t) \cap \mathcal{B}(H_a),
    \end{array}
\end{equation}
where $U(t) = [u(t|t)^\top, u(t+1|t)^\top, \dots, u(t+N-1|t)^\top]^\top$, and the regressor $\Phi(k|t)$ is as in \eqref{phi_hor}. We have included the terminal constraint on the regressor vector as given in \cite{tanaskovic2014adaptive}
\begin{align}\label{eq:trc}
\Phi(t+N|t)=W\Phi(t+N|t)+Z u(t+N-1|t).
\end{align}
This means the terminal regressor corresponds to a steady state (i.e. the last $m$ control inputs in a horizon are kept constant). Problem \eqref{mpc_problem} is a convex optimization problem and can be solved  with existing solvers. After solving \eqref{mpc_problem}, we apply the first input 
\begin{equation}\label{eq:RHC}
    u(t)=u^\star(t|t)
\end{equation}
to system \eqref{sysmodel} in closed-loop. We then re-solve \eqref{mpc_problem} at timestep $t+1$ with new estimated data $\mu_a(t+1)$ and $\mathcal F(t+1)$. This yields a receding-horizon control scheme. The resulting algorithm is summarized in Algorithm~\ref{alg1}. 
\begin{algorithm}
    \caption{
Adaptive Stochastic MPC: Sparse-FIR MIMO Systems
    }
    \label{alg1}
  \begin{algorithmic}[1]
      \Statex \hspace{-1.2em}\textbf{Initialize:} $\mathcal{F}(0), \mu_\mathbf{a}(0), \sigma_\mathbf{a}(0)$.
      \Statex 
      \hspace{-1.2em}\textbf{Inputs:} 
      $q, \bar{w}, \bar{k}, t_\mathrm{end}$ 
      \vspace{1.2mm}
      \Statex \hspace{-1.2em} \emph{begin Basis Pursuit Denoising (offline)} 

      \STATE Construct offline regressors $[\Phi_1,\Phi_2,\dots,\Phi_q]$ such that operator $\mathcal{A}$ in Proposition~\ref{prop:sparse_bound} satisfies $\delta_{2\bar{k}}(\mathcal{A}) < \bar{\delta} \ll \sqrt{2}-1$; 
      
      \STATE Solve \eqref{eq:sparse_recover} for $i=1,2,\dots,n_y$ to obtain the FSPS $\mathcal{B}(H_a)$;

      \Statex \hspace{-1.2em} \emph{end Basis Pursuit Denoising}  
      \vspace{1.2mm}
      
      \Statex \hspace{-1.2em} \emph{begin MPC control process (online)} 
      \FOR{timestep $1 \leq t \leq t_\mathrm{end}$}
      
      \STATE Obtain $y(t)$ and update the FPS $\mathcal{F}(t)$ using \eqref{support};

      \STATE Estimate mean and variance $\mu_{\mathbf{a}}(t)$ and $\sigma^2_{\mathbf{a}}(t)$ with RLS estimator. Project the mean 
      
      \hspace{-1.2em} with \eqref{eq:proj} to the set $\mathcal{F}_p(t)$;
      
      \STATE Solve \eqref{mpc_problem} and apply $u(t) = u^\star(t|t)$ to system \eqref{sysmodel};
      \ENDFOR
      \end{algorithmic}
\end{algorithm}

\begin{theorem} 
Consider Algorithm~1 and the receding horizon closed-loop control law \eqref{eq:RHC} applied to system \eqref{sysmodel} after solving optimization problem \eqref{mpc_problem}. If the optimization problem \eqref{mpc_problem} is feasible at timestep $t=0$, then it is feasible at all subsequent timesteps $0 \leq t \leq t_\mathrm{end}$. 
\end{theorem}

\begin{proof} Let \eqref{mpc_problem} be feasible at timestep $t$ and let the corresponding optimal input sequence be $U^\star(t) = [u^\star(t|t)^\top,u^\star(t+1|t)^\top,\dots,u^\star(t+ N-1|t )^\top]^\top$. After applying control \eqref{eq:RHC} in closed-loop to \eqref{sysmodel}, consider a candidate open loop control sequence at the next timestep as
\begin{align} \label{feasinput}
   U(t+1) = [u^\star(t+1|t)^\top,\dots,u^\star(t+ N-1|t )^\top, u^\star(t + N-1|t)^\top]^\top.
\end{align}
This candidate sequence \eqref{feasinput} satisfies the input constraints \eqref{contr_con} and the terminal constraint \eqref{eq:trc} at $(t+1)$. So, using candidate sequence \eqref{feasinput} and condition \eqref{eq:trc}, we obtain
\begin{subequations} \label{phic}
\begin{align} 
& \Phi(k|t +1) = \Phi^\star(k|t), \forall k\in \{t+2,\dots,t+N\},\\
& \Phi(t +N +1|t +1) = \Phi^\star(t +N|t).
\end{align}
\end{subequations}
To show feasibility of \eqref{mpc_problem} at $(t+1)$, we finally must ensure that \eqref{socp} satisfied with \eqref{feasinput} at $(t+1)$. That is, we require for all $k\in \{t+2,\dots,t+N+1\}$
\begin{align}\label{rfc}
    \kappa_{\epsilon} \sqrt{{\bar{\Phi}}^\top(k|t+1)\Gamma{\bar{\Phi}}(k|t+1)} + \Phi^\top(k|t+1)\bar{E}f^j(t+1)-p \leq 0,
\end{align}
for all $f^j(t+1) \in \mathcal{F}(t+1)$ vertices. Now, for the chosen input sequence (\ref{feasinput}), by using (\ref{eq:trc}) and (\ref{phic}), condition (\ref{rfc}) can be expressed as
\begin{equation}\label{rfc_final_socp}
    \kappa_{\epsilon} \sqrt{{\bar{\Phi}}^{\star\top}(k|t)\Gamma{\bar{\Phi}}^\star(k|t)} + \Phi^{\star\top}(k|t)\bar{E}f^j(t+1) -p \leq 0,
\end{equation}
for all $k \in \{t+2,\dots,t+N\}$. We can now guarantee that \eqref{rfc_final_socp} will be satisfied at timestep $(t+1)$ with \eqref{feasinput}. This is due to the observation that feasible parameter set follows $\mathcal{F}(t) \supseteq \mathcal{F}(t+1)$ as new cuts (\ref{support}) are introduced at each timestep. So vertices $f^j(t+1)$ at timestep $(t+1)$ are convex combinations of the ones at $t$. Thus, control sequence \eqref{feasinput} is feasible for \eqref{mpc_problem} at timestep $(t+1)$. This completes the proof.
\end{proof}

\section{Numerical Simulations}\label{sec:simul}
We present simulation results for a simple single-input, single-output system\footnote{We choose this model for the sake of simulation simplicity. Sparsity in FIR models is well motivated typically for MIMO systems.}. We compare the performance of our Algorithm~\ref{alg1} with that from the adaptive stochastic MPC presented in \cite{bujarbaruahAdapFIR}. This performance is measured in terms of the expected closed-loop cost $\mathbb{E}[\mathcal{V}]$, where the closed-loop cost of any trajectory which is a function of the disturbance realization $\bar{\mathbf{w}} = [w(0),w(1),\dots,w(t_\mathrm{end})]$, is given by 
\begin{equation*}
    \mathcal{V}(\bar{\mathbf{w}}, \Phi(0), \mathcal{F}_p(0), \mu_\mathbf{a}(0), \sigma_\mathbf{a}(0)) = \sum \limits_{t=0}^{t_\mathrm{end}} y^\top(t) Q y(t) + (u^\star(t))^\top S u^\star(t),
\end{equation*}
where $\mathcal{F}_p(\cdot)$ for Algorithm~\ref{alg1} is obtained as in \eqref{eq:point_estimDom}, and for the algorithm in \cite{bujarbaruahAdapFIR}, $\mathcal{F}_p(\cdot) = \mathcal{F}(\cdot)$, as Remark~\ref{rem:diff_with_acc18} points out. For simulating both the algorithms, we use the parameters given in Table~\ref{table:par}, with the true system response given as $H_a = [-1, 0, 0, 0, 0, 0, 0, 0, -2, 0]$, which is $\bar{k} = 2$ sparse. 
    \begin{table}[!h]
	\renewcommand{\arraystretch}{1}
	\caption{Simulation Parameters}
	\vspace{0.5cm}
	\label{table:par}
	\centering
	 \begin{tabular}{| c c | c c |}
		\hline
		\bfseries Parameter & \bfseries Value & 
		\bfseries Parameter & \bfseries Value \\
		\hline
		$m$ & $10$ & $N$ & $12$ \\
		$t_\mathrm{end}$ & $20$ & $w$ & $\textit{U}\sim[-0.1,0.1]$  \\
		$n_u$ & $1$ & $n_y$ & $1$ \\
		$\epsilon$ & $0.1$ & $\kappa_\epsilon$ & $3$\\
		$E$ & $1$ & $p$ & $5$ \\
		$C$ & $\mathrm{diag}(1,-1)$ & $g$ & $[1,1]^\top$\\ 
		$Q$ & $\mathrm{diag}(20,20)$ & $S$ & $\mathrm{diag}(2,2)$\\ 
		$\mu_{\mathbf{a}}(0)$ & $\mathbf{1}_{10 \times 1}$ & $\sigma^2_{\mathbf{a}}(0)$ & $0.1 \times \mathbb{I}_{10}$\\ 
		$\Phi(0)$ & $0.1\times \mathbf{1}_{10\times 1}$ & $\mathcal{F}(0)$ & $-\mathbf{3}_{10\times 1} \leq H^\top \leq \mathbf{3}_{10\times 1}$\\ 
		\hline
	\end{tabular}
\end{table}
We run $100$ Monte Carlo simulations with both algorithms for $100$ randomly chosen disturbance sequences $\bar{\mathbf{w}}$. We approximate the average closed-loop cost $\mathbb{E}[\mathcal{V}]$ with the empirical average
\begin{align}\label{eq:emp_MCcost}
    \hat{\mathcal{V}}(\Phi(0),\mathcal{F}_p(0), \mu_\mathbf{a}(0), \sigma_\mathbf{a}(0)) = \frac{1}{100} \sum_{\tilde{m}=1}^{100} \mathcal{V}((\bar{\mathbf{w}})^{\star \tilde{m}}, \Phi(0),\mathcal{F}_p(0), \mu_\mathbf{a}(0), \sigma_\mathbf{a}(0)),
\end{align}
where $(\cdot)^{\star \tilde{m}}$ represents the $\tilde{m}^{\mathrm{th}}$ Monte Carlo sample. 

\subsection{Cost Comparison}
Fig.~\ref{fig:cost} shows the comparison of closed-loop cost expressed as $\mathcal{V}(\bar{\mathbf{w}}, \Phi(0), \mathcal{F}_p(0), \mu_\mathbf{a}(0), \sigma_\mathbf{a}(0)) = \sum \limits_{t=0}^{20} y^\top(t) Q y(t) + (u^\star(t))^\top S u^\star(t)$ for $100$ different Monte Carlo draws of trajectories, obtained with Algorithm~\ref{alg1} and the algorithm in \cite{bujarbaruahAdapFIR}. We see that the empirical average closed-loop cost obtained as \eqref{eq:emp_MCcost} for Algorithm~\ref{alg1} is around $30\%$ lower than the corresponding value obtained with \cite{bujarbaruahAdapFIR}. 
\begin{figure}[h]
    \centering
    \includegraphics[width=12cm]{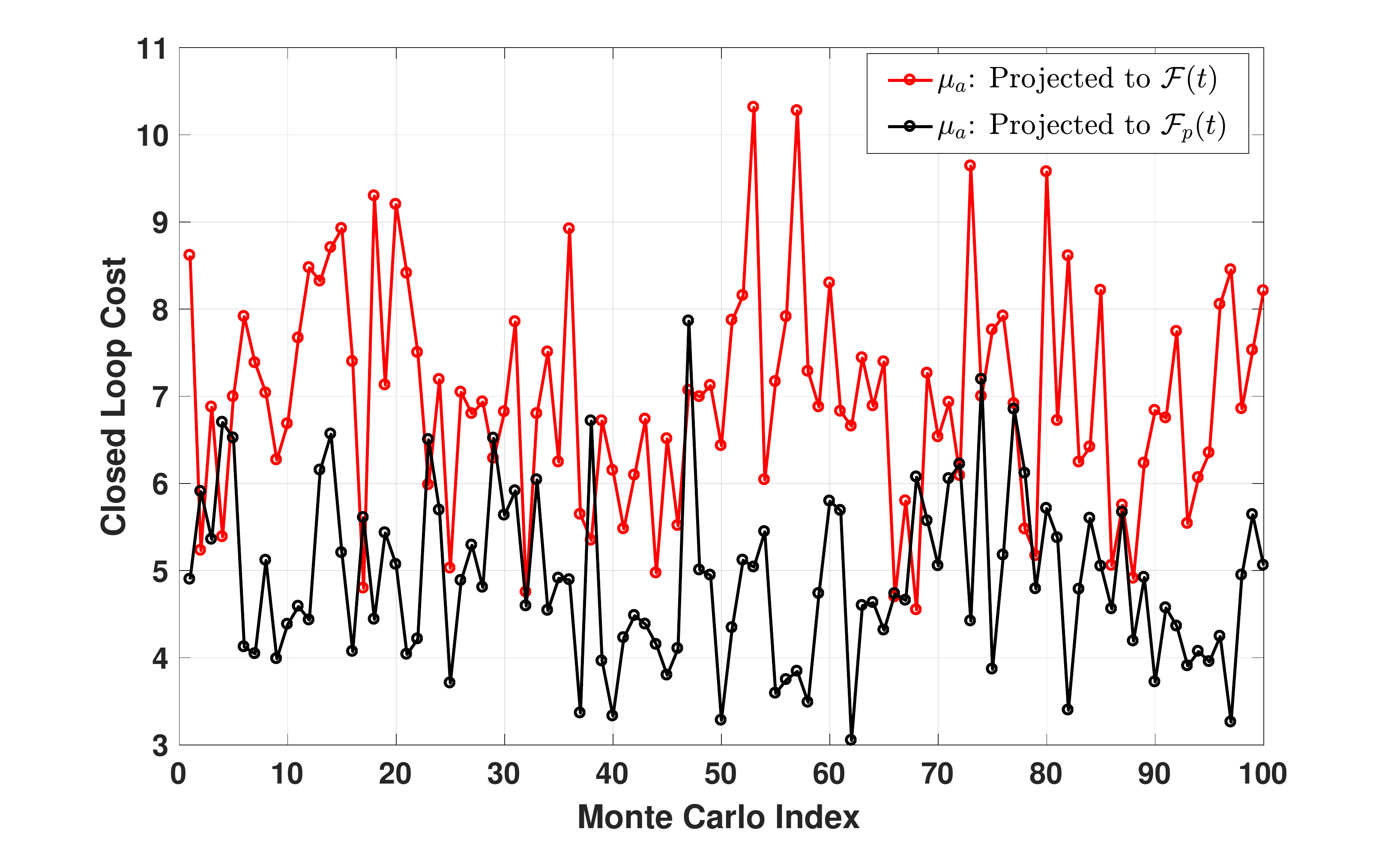}
    \caption{Closed-Loop Cost $\sum \limits_{t=0}^{20} y^\top(t) Q y(t) + (u^\star(t))^\top S u^\star(t)$ Along 100 Monte Carlo Simulation of Trajectories.}
    \label{fig:cost}
\end{figure}
This demonstrates performance gain by Algorithm~\ref{alg1} as a consequence of leveraging sparsity information of $H_a$ via the FSPS set $\mathcal{B}(H_a)$. Furthermore, the closed-loop costs obtained for $90\%$ out of the $100$ trajectories were lower with Algorithm~\ref{alg1} compared to the corresponding cost obtained with \cite{bujarbaruahAdapFIR}. This improvement is cost is explained in the next section via a closer look at the corresponding trajectories of the system.

\subsection{Closed-loop Trajectory Comparison}
In order to further analyze the performance gain seen above, we consider the difference in absolute value of the outputs along each one of 100 trajectories (Fig.~\ref{fig:out}). Specifically, we analyze
\begin{align*}
    \Delta \vert y(t) \vert = \vert y^1(t) \vert - \vert y^2(t) \vert,  
\end{align*}
for $t \in [0,20]$, where $y^2(\cdot)$ and $y^1(\cdot)$ denote the outputs obtained using Algorithm~\ref{alg1} and the algorithm in \cite{bujarbaruahAdapFIR} respectively. We see from Fig.~\ref{fig:out} that out of all $2000$ timesteps simulated, only at about $30\%$ of them is $\Delta \vert y(t) \vert <0$. This implies that at about $70\%$ of all timesteps simulated, Algorithm~\ref{alg1} resulted in outputs closer to the origin, as desired.  
\begin{figure}[h]
    \centering
    \includegraphics[width=12cm]{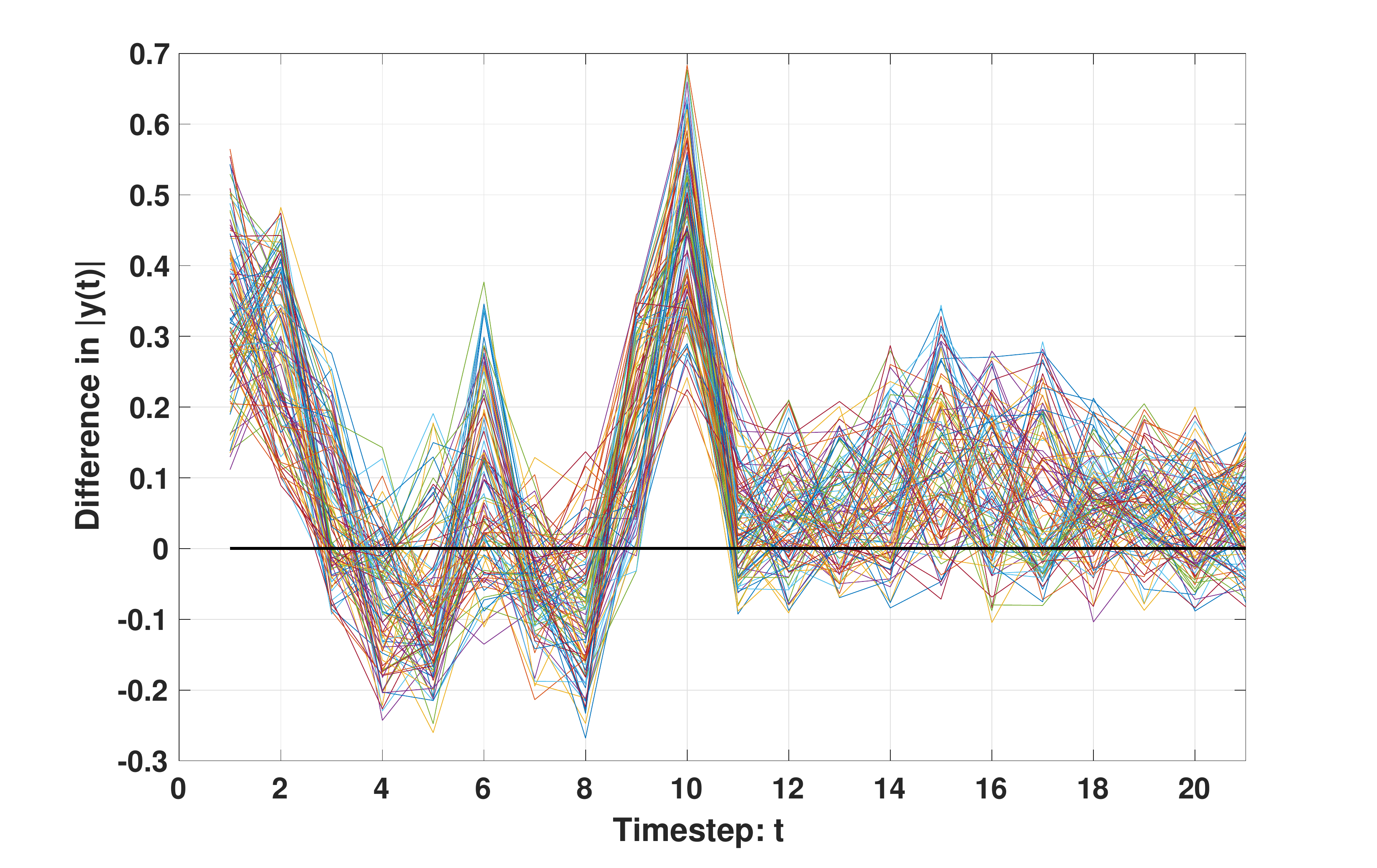}
    \caption{Difference in Closed-Loop Outputs $\Delta \vert y(t) \vert$ Along 100 Monte Carlo Simulation of Trajectories.}
    \label{fig:out}
\end{figure}

Fig.~\ref{fig:in} further plots the difference in absolute value of corresponding closed-loop inputs, i.e., 
\begin{align*}
\Delta \vert u^\star(t) \vert = \vert u^{\star,1}(t) \vert - \vert u^{\star,2}(t) \vert, 
\end{align*}
for $t \in [0,20]$ along all the 100 trajectories, where $u^{\star,2}(\cdot)$ and $u^{\star,1}(\cdot)$ denote the optimal closed-loop inputs obtained using Algorithm~\ref{alg1} and the algorithm in \cite{bujarbaruahAdapFIR} respectively. At only $19\%$ of the 2000 timesteps  $\Delta \vert u^\star(t) \vert <0$, which indicates that the improved output regulation in Fig.~\ref{fig:out} does not come at the cost of higher magnitudes of inputs. Thus, Fig.~\ref{fig:out} and Fig.~\ref{fig:in} combined, provides a comprehensive explanation of cost improvement in Fig.~\ref{fig:cost}. 
\begin{figure}[h]
    \centering
    \includegraphics[width=12cm]{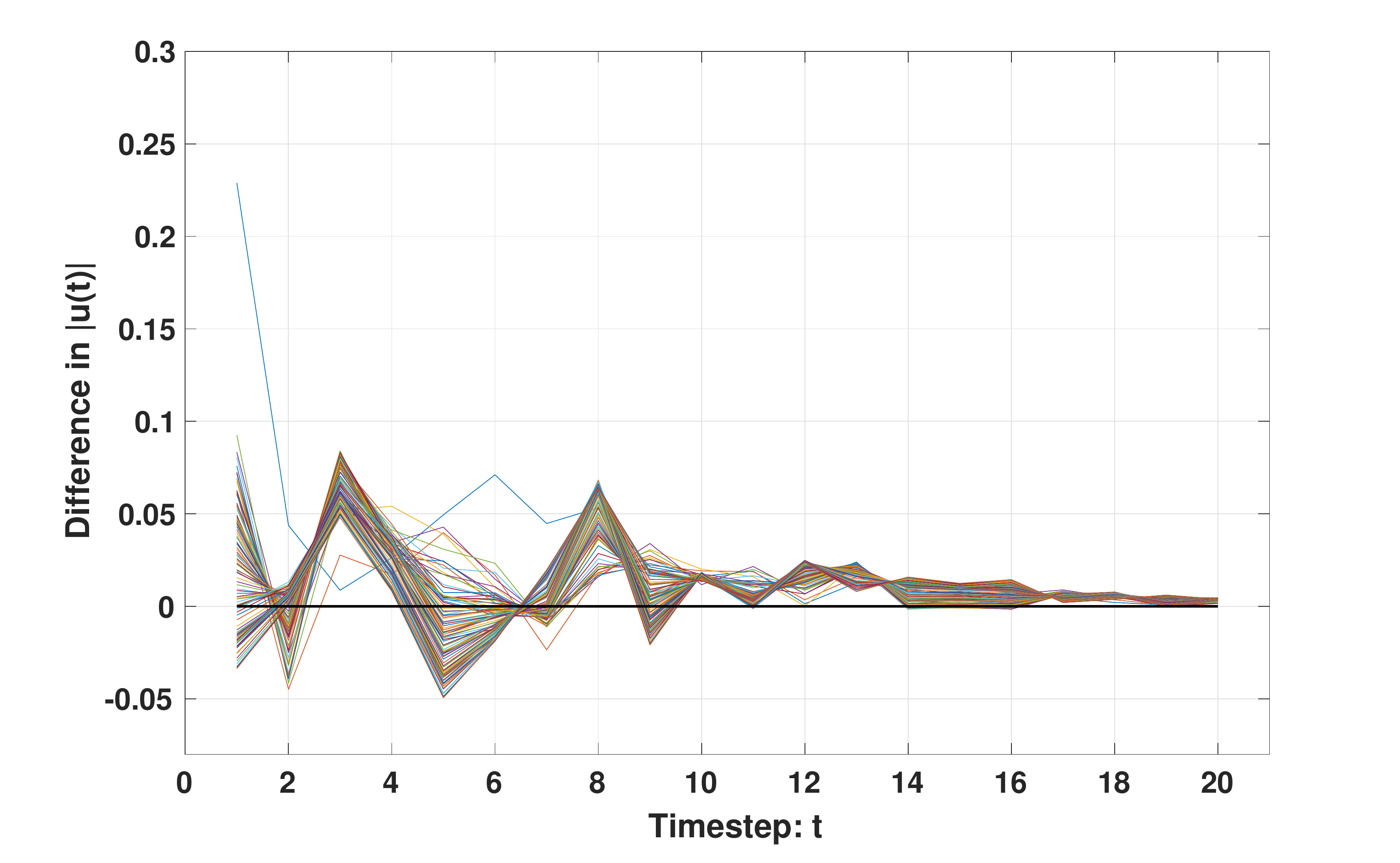}
    \caption{Difference in Optimal Closed-Loop Inputs $\Delta \vert u^\star(t) \vert$ Along 100 Monte Carlo Simulation of Trajectories.}
    \label{fig:in}
\end{figure}

\section{Conclusions}
We proposed an Adaptive Stochastic Model Predictive Control (MPC) algorithm for stable linear time-invariant MIMO systems expressed with an FIR model. The parameters of the FIR model corresponding to each output are sparse, but unknown.  We solve a Basis Pursuit Denoising problem offline before the control process, which utilizes the sparsity information of the system model to give set based bounds containing the true system parameters. Online during control process, we estimate these parameters with the Recursive Least Squares algorithm and the estimates are refined using the set based bounds obtained offline. 
With Set Membership Methods, our MPC controller safely ensures satisfaction of all imposed constraints by the unknown true system.
With a thorough numerical example we demonstrated performance gain over the algorithm of \cite{bujarbaruahAdapFIR}.

\section*{Appendix}
\subsection*{Choosing Offline Regressors}
Following \cite[Theorem~3.3.4]{YiMaBook}, there exists a numerical constant $\tilde{C} > 0$ such that if $\mathcal{A}= [\Phi_1, \Phi_2, \dots, \Phi_q]^\top \in \mathbb{R}^{q \times n_um}$ is a random matrix with entries independent $\mathcal{N}(0,\frac{1}{q})$ random variables, with high probability, $\delta_{2\bar{k}}(\mathcal{A}) < \bar{\delta} \ll \sqrt{2} -1$, provided 
\begin{align}\label{eq:q_count}
q \geq \frac{2\tilde{C}\bar{k} \log(\frac{n_um}{2\bar{k}})}{\bar{\delta}^2}. \end{align}
Offline regressor vectors $[\Phi_1, \Phi_2, \dots, \Phi_q]$ are thus chosen satisfying \eqref{eq:q_count}, with each entry of $\mathcal{A}$ as $\mathcal{N}(0,\frac{1}{q})$. The constant $\tilde{C}$ can be exactly found using properties of phase transition in compressed sensing \cite{donoho2006compressed}. 

\subsection*{Matrix Notations}
We define
\begin{align*}
    &\mathbf{H}(t) = [H_{1}, H_{2}, \dots, H_{n_y}]^\top \in \mathbb{R}^{n_yn_um \times 1},\\
    &\bm{\Phi}(t)  = \mathrm{diag}(\Phi^\top(t),\Phi^\top(t),\dots,\Phi^\top(t))\in \mathbb{R}^{n_y \times n_yn_um},
\end{align*}
where $H_i(t)$ denotes the $i^\mathrm{th}$ row of $H(t)$. Moreover, 
\begin{align*}
\bar{q} & = \left[\begin{array}{ccccc}
0&0&\cdots &0&0\\
1&0&\cdots &0&0\\
0 &1 &\cdots &0&0\\
\vdots &\vdots &\ddots &\vdots&\vdots\\
0&0&\cdots &1&0\\
\end{array}\right] \in \mathbb{R}^{m \times m}
\end{align*}
Based on this matrix $\bar{q}$ we get: 
\begin{align*}
& W  =  \mathrm{diag}(\bar{q},\bar{q},\dots,\bar{q}) \in \mathbb{R}^{n_um \times n_um},\\
& z = [1, 0, \dots, 0]^\top \in \mathbb{R}^m, \end{align*}
which gives $Z=  \mathrm{diag}(z,z,\dots,z) \in \mathbb{R}^{n_um \times n_u}$.

\subsection*{Chance Constraint to Conic Optimization}
For this part we repeat the derivation of \cite{bujarbaruahAdapFIR}. To ensure satisfaction of \eqref{eq:cc1} with unknown true system $H_a$, we must satisfy them for all $H(t) \in \mathcal{F}(t)$. Thus $H(t)$ is treated as a deterministic variable (mean $H(t)$ and variance $\mathbf{0}$) while reformulating \eqref{eq:cc1} to equivalent convex constraints. Therefore, for all $H(t) \in \mathcal F(t)$, we have:
\begin{align*}
\mathbb{P}\{EH(t)\Phi(k|t) + Ew(k) \leq p\} & \geq 1-\epsilon\\
\iff \mathbb{P}\{[EH(t) \quad Ew(k)][\Phi^\top(k|t) \quad 1]^\top \leq p\} & \geq 1-\epsilon. 
\end{align*}
We denote, 
\begin{align*}
 a_1^\top(t) & = [EH(t) \quad Ew(k)], \implies \hat{a}_1^\top(t) = [EH(t) \quad 0], 
\end{align*}
where $\hat{x}$ denotes the mean of a quantity $x$. Moreover,
\begin{align*}
 \bar{\Phi}(k|t) & = [\Phi^\top(k|t) \quad 1 \quad 1]^\top,~\textrm{and,}
\end{align*}
\begin{align*}
 d_1(t) & = [a_1^\top(t) \quad -p]^\top, \implies \hat{d}_1(t) = [\hat{a}_1^\top(t) \quad -p]^\top.
\end{align*}
From these we can derive the variance $\Gamma$ of $d_1(t)$ as
\begin{align*}
     \Gamma & = \sigma^2(d_1(t)) = \sigma^2(\begin{bmatrix}
         \bar{E}{\mathbf{H}(t)} \\
         E {w}(k)\\
         -p\\
        \end{bmatrix}),~\text{where }
     \bar E = \mathrm{diag}(E,E,\dots,E),\\
& = \mathrm{diag}(0, E\sigma^2_wE^\top,0) \succeq 0.
\end{align*}
As in \cite{bujarbaruahAdapFIR}, we assume no correlation between the disturbance and the impulse response distribution. Now, \eqref{eq:cc1} is reformulated into second-order cone constraints following \cite{calafiore2006distributionally}
\begin{equation}\label{socp11}
\kappa_{\epsilon} \sqrt{\{{\bar{\Phi}}^\top(k|t)\Gamma{\bar{\Phi}}(k|t)}\} + \hat{d}_1^\top(t){\bar{\Phi}}(k|t) \leq 0,
\end{equation}
for all $k\in \{t+1,\dots,t+N\}$, where $\kappa_\epsilon = \sqrt{\frac{1-\epsilon}{\epsilon}}$ for any bounded disturbance distributions $w(k)$ with known moments. After simplifications, \eqref{socp11} can be written for all $H(t) \in \mathcal F(t)$ as
\begin{equation*}
\kappa_{\epsilon} \sqrt{\{{\bar{\Phi}}^\top(k|t)\Gamma{\bar{\Phi}}(k|t)}\} + \Phi^\top(k|t)\bar{E}\mathbf{H}(t) -p \leq 0.
\end{equation*}
\bibliographystyle{IEEEtran} 
\bibliography{bibliography.bib}

\end{document}